\documentclass[12pt,reqno]{amsart}
\usepackage{amsmath,amsfonts,amsthm,amssymb,color,hyperref}

\usepackage{cases}
\usepackage{float}

\usepackage[T1]{fontenc}
\usepackage{subfig} 
\usepackage{graphicx}

\newtheorem*{theorem*}{Theorem}
\newtheorem{proposition}{Proposition}[section]

\theoremstyle{definition}

\newtheorem{remark}[proposition]{Remark}

\def \E {\mathbb{E}}

\begin{document}


\title[Shafik et~al.]{Flexible transition probability
model for assessing cost-effectiveness of breast cancer screening}

\author[Shafik]{Nourhan Shafik}
\address{Department of Mathematics and Systems Analysis, Aalto University School of Science, Finland \\ 
Institute for Statistical and Epidemiological Cancer Research, Finnish Cancer
Registry, Helsinki, Finland} 
\author[Ilmonen]{Pauliina Ilmonen}
\address{Department of Mathematics and Systems Analysis, Aalto University School of Science, Finland}
\email{pauliina.ilmonen@aalto.fi}
\author[Viitasaari]{Lauri Viitasaari}
\address{Uppsala University, Department of Mathematics, Uppsala, Sweden}
\author[Sarkeala]{Tytti Sarkeala}
\address{Institute for Statistical and Epidemiological Cancer Research, Finnish Cancer
Registry, Helsinki, Finland} 
\author[Hein\"avaara]{Sirpa Hein\"avaara}
\address{Institute for Statistical and Epidemiological Cancer Research, Finnish Cancer
Registry, Helsinki, Finland}

\begin{abstract}
Breast cancer is the most common cancer among Western women. Fortunately, organized screening has reduced breast cancer mortality and, consequently, the European Union has recommended screening with mammography for 50–69-year-old women. This recommendation is followed well in Europe. Widening the screening target age further is supported by conditional recommendations for 45–49- and 70–74-year-old women. However, before extending screening to new age groups, it's essential to carefully consider the benefits and costs locally as circumstances vary between different regions and/or countries. We propose a new approach to assess cost-effectiveness of breast cancer screening for a long-ongoing program with incomplete historical screening data. The new model is called flexible stage distribution model. It is based on estimating the stage distributions of breast cancer cases under different screening strategies. In the model, an ongoing screening strategy may be used as a baseline and other screening strategies may be incorporated by changes in the incidence rates. The model is flexible, as it enables to apply different approaches for estimating the altered stage distributions. Thus, if randomized data is available, one may rely on that. On the other hand, if randomized data is not available, altered stage distributions may be estimated by extrapolating the stage distributions of the youngest and oldest screened/non-screened age groups. We apply the proposed flexible stage distribution model for assessing incremental cost of extending the current biennial breast cancer screening to younger and older target ages in Finland.
\end{abstract}

\keywords{cancer screening, breast cancer, cost effectiveness, stage distribution}

\maketitle

\section{Introduction}
Breast cancer (BC) is the most common cancer among Western women \cite{IARC}. Fortunately, mortality due to BC can be and have been successfully reduced by organized screening \cite{UK2012, Euroscreen2012, Nelson2016}. Subsequently, the EU has recommended screening with mammography for 50–69-year-old women and this has been adapted widely in Europe \cite{EU_guideline2013, ECIBC, Basu2018}. Potential to widen the target age further is supported by conditional recommendations for 45–49- and 70–74-year-old women \cite{ECIBC}. However, prior to possible widening, it's essential to carefully consider the benefits and costs locally as circumstances vary between different regions and/or countries. The assessment of benefits, breast cancer mortality reduction or life-years gained (LYG) against costs, should take into account country-specific conditions, and provide support for an optimal use of available health care resources.

Cost-effectiveness modelling has been recommended to support policy-making for a new or a recently started screening program \cite{CanCon}. Usually, in cost-effectiveness modelling, alternative screening strategies are assumed to occur in the future for a cohort's lifetime and no-screening is used as a reference. From modelling perspective, ideally, one could rely on recent data where individuals were randomly assigned to different screening strategy groups, including the no screening group. In practice this is seldom the case and lack of data may severely complicate cost-effectiveness modelling. The registration of BC screening may also have been incomplete, especially if screening has started decades ago. Moreover, we have seen significant changes in both, incidence and treatment guidelines. Thus, lack of suitable data makes standard cost-effectiveness modelling unreliable.
Cost-effectiveness modelling is indeed challenging if a screening program has been ongoing for decades and its previous performance is not fully known. Indeed, reliability of future predictions depends on the past, on a prediction base. If underlying data are incomplete or heterogeneous, one cannot lay high confidence on future predictions either. A corresponding challenge occurs if a situation without screening has existed several decades ago. No-screening is then not the most relevant reference. Instead, it would be reasonable to compare alternative screening strategies primarily to the current screening.

We propose a new approach to assess cost-effectiveness of BC screening for a long-ongoing program with incomplete historical screening data. The new model is called flexible stage distribution model. It is based on estimating the stage distributions of BC cases under different screening strategies. In the model, an ongoing screening strategy may be used as a baseline and other screening strategies may be incorporated by changing the incidence rates. The model is flexible, as it enables to apply different approaches for estimating the altered stage distributions. Thus, if randomized data is available, one may rely on that. On the other hand, if randomized data is not available, altered stage distributions may be estimated by extrapolating the stage distributions of the youngest and oldest screened/non-screened  age groups. Moreover, the model enables to use either TNM (tumor (T), nodes (N), and metastases (M)) or some other cancer classification depending on availability of data.

We apply the proposed flexible stage distribution model for assessing incremental cost of extending the current biennial breast cancer screening to younger and older target ages in Finland.

\subsection{Ethical considerations}
Permit for this register-based study has been granted by Finnish Social and Health Data Permit Authority Findata (THL/504/14.06.00/2020). 

\section{Breast cancer screening in Finland}
In Finland the nationwide organized screening started in 1992 targeted biennially to 50-59-year-old women. Individual municipalities are responsible to offer screening to their permanent residents and have often offered breast screening also to wider target ages \cite{Sarkeala}. Screening has therefore been ongoing with varying target ages by time and municipalities. The current nationwide target age, 50-69 years, was introduced gradually and adopted fully in 2017.
Participation to screening is currently about 82\% and varies very little within the target age \cite{Screeningstatistics}. Registration of screening data has been complete since 2000. 

\section{Flexible transition probability
model}
The proposed flexible transition probability model is based on modeling the effect of screening on cancer stage distributions at the time of the first diagnosis. This is done separately for different age groups. Costs of treatment and survival depend on the stage distribution and the age group.

Age groups are indexed by $j=1,2,\ldots,J$, and screening rounds are taken at every time instance $j$. In our study, we consider biennial screening, and the first age group $j=1$ corresponds to the 46-47 year old females, the second to the 48-49 year old females, and the last one to the 98-99 old  females. For simplicity, we assume that no one survives over 100 years. The age groups and screening intervals in our study are chosen to match available data and current policies in Finland. We consider different policies $h =[h(1), h(2), \ldots, h(J)]\in \{S,NS\}^J$ where $S$ corresponds to the screening and $NS$ to no-screening. For example, with two age-groups, i.e. $J=2$, a policy $h =[S,S]$ corresponds to screening of both age groups, $h=[NS,S]$ to screening of older age group, and so on. In each of the screening round, individuals are diagnosed with stage $k$ that are indexed by $k=-1,0,\ldots, K$, with $k=-1$ corresponding to no breast cancer. In our study, we have $K=4$, where the stages correspond to 0 - Unknown, 1 - Localized, 2- Non localized/Regional lymph nodes metastases, 3 - Metastasized farther than regional lymph nodes or invades adjacent tissues, and 4 - In situ carcinoma. Alternatively, one could, for example, use TNM classifications here. That was not possible in our study as Finnish Cancer Registry does not register TNM classifications. The observed state of an individual $i$ of age $j$ is denoted by $X^i_{j,h}$. Note that the observed state depends on the policy $h$ through the value $h(j)$ determining whether the age group $j$ is screened or not. We denote by $\mu^i_{j,h} = \mu_{j,h}$ the state distribution of $X$, observed at screening. That is, $\mu_{j,h}(k) = P(X_{j,h}^i = k)$. Note that here it is assumed that $\mu$ does not depend on the particular individual (only the age), but it depends on the policy through value $h(j)$ representing the choice whether age group $j$ is screened or not. For this reason, we omit the superscript $i$ in the notation and simply write, e.g. $X_{j,h}$ whenever confusion may not arise.
\begin{remark}
For our purposes, we model the conditional distribution $P(X^i_{j,h} = k | X^{i}_{j,h} \neq -1)$ and incidence rates separately. That is, we model separately the probabilities of given stages conditioned on breast cancer being diagnosed, and the incidence rate $P(X^i_{j,h} \neq -1)$. The connection to actual stage distribution $\mu_{j,h}$ is simply
{\footnotesize
$$
\mu_{j,h}(k) = P(X^{j,h} = k)= P(X^i_{j,h} = k | X^{i}_{j,h} \neq 0)P(X^{i}_{j,h} \neq 0).
$$}
\end{remark}

Once an individual $i$ of age $j$ is diagnosed with stage $k$, the time to death is denoted by $T^i_{j,k}$ and its distribution, not depending on $i$, is denoted by $\lambda_{j,k}$. That is, we have $\lambda_{j,k}(t) = P(T^i_{j,k} = t)$. Hence the expected number of years left, once the individual $i$ of age $j$ is on stage $k$, is
\begin{equation}
\label{eq:expected-years}
\E T^i_{j,k} = \sum_{t=0}^\infty t\lambda_{j,k}(t).
\end{equation}
We are interested in the effect of screening only and individuals are removed from the screening population once they either die (for any reason) or are diagnosed with breast cancer. We model a cohort of size $N_0 = 100000$ of individuals throughout their life span, starting at the first (possible) screening age. The dynamics of the screening population $N_j$ is hence given by
{\footnotesize$$
N_{j+1} = N_j \left(1-P(T_{j}=0, T_j=1) - P(T_j\geq 2,X_{j,h} \neq -1)\right).
$$}
Here $T_j$ denotes the time to death of an individual at age $j$, regardless of the stage. That is, at each screening step (age $j$) we invite $N_j$ individuals to screening, and the amount $N_{j+1}$ invited to the next screening contains $N_j$ from which we have removed those who have died during the two years (two year screening interval) with proportion $P(T_j=0, T_j=1)$ and those who survive at least two years for the next screening to happen, but got diagnosed, i.e. the stage $X\neq -1$. These individuals have proportion $P(T_j\geq 2, X_{j,h}\neq -1)$. Clearly, we have
$$
P(T_j=t) = \sum_{k=-1}^K P(T_{j,k}=t)P(X_{k,h}=k)
$$
$$= \sum_{k=-1}^K \lambda_{j,k}(t)\mu_{j,h}(k).
$$
The total number of life years left, measured at the beginning, of the entire population is given by the following result.
\begin{proposition}
\label{prop:years}
The total number of life years left $T_h$ with a given policy is given by
{\small$$
\E T_h = \sum_{j=2}^J 2N_j + \sum_{j=1}^J \left[\lambda_{j,-1}(1)\mu_{j,h}(-1) + \sum_{k=0}^K \mu_{j,h}(k)\E T_{j,k}\right]N_j,
$$}
where $\E T_{j,k}$ is given by \eqref{eq:expected-years}.
\end{proposition}
\begin{proof}[Proof of Proposition \ref{prop:years}]
On each screening round, we invite $N_j$ individuals to screening. The proportion of those who are not diagnosed with breast cancer is $\mu_{j,h}(-1)$, and the proportion of age $j$ at stage $k=-1$ (no breast cancer) live only one year is $\lambda_{j,-1}(1)$. These individuals have one remaining year, contributing a factor $\lambda_{j,-1}(1)\mu_{j,h}(-1) N_j$. Similarly, the individuals diagnosed with breast cancer ($k\neq -1$) have expected remaining years $\E T_{j,k}$, and the proportion of individuals is $\mu_{j,h}(k)$. Summing over different stages $k\neq -1$ contributes with factor $\sum_{k=0}^K \mu_{j,h}(k)\E T_{j,k}N_j$. Finally, the amount of individuals who continue to the next screening round $j+1$ is $N_{j+1}$, who at age $j+1$ are then contributed with two lived years. Summing over screening rounds (i.e., over age groups $j$) contributes with the factor
$$
\sum_{j=2}^J 2 N_j.
$$
This completes the proof.
\end{proof}

Consider next the related cancer costs. The cost of treatment and death of an individual $i$ of age $j$ depends on the stage $k$ (given by the variable $X^i_{j,h}$), the time to death $t$ (given by the variable $T^{i}_{j,k}$), the age $j$, and the cause of death. We denote by $d=1,2,\ldots, D$ the different causes of death, and for an individual $i$ of age $j$ in stage $k$, the random variable $D^{i}_{j,k} \in \{1,\ldots, D\}$ determines the cause of death. In our study, we use $D=2$ and consider deaths due to breast cancer ($d=1$) or due to other causes ($d=2$). Given the quadruple $(j,k,t,d)$, the average cost $\tilde{C}_{j,k,t,d}$ is considered as constant, and consists of treatment related costs solely. That is, if $k=-1$, then we obtain no costs as no treatment is required. Hence, the treatment and death cost $C^i_{j,k}$ of an individual $i$ of age $j$ and at state $k$ is given by
$$
C^{i}_{j,k} = \sum_{t=0}^\infty \sum_{d=1}^D \tilde{C}_{j,k,t,d} 1_{T^i_{j,k} = t, D^{i}_{j,k} = d},
$$
where $1_{T^i_{j,k} = t, D^{i}_{j,k} = d}$ denotes the indicator taking value $1$ if the individual $i$ at age $j$ lives $t$ years after diagnosed with stage $k$ and dies to the cause $d$. We denote by $\pi_{j,k}$ the joint distribution of $T^{i}_{j,k}$ and $D^{i}_{j,k}$ that is assumed to be independent of the chosen individual (of age $j$ and at stage $k$). That is, we have
$$
\pi_{j,k}(t,d) = P(T^i_{j,k}=t, D^i_{j,k}=d),
$$
giving the probability that an individual at age $j$ and stage $k$ lives exactly $t$ years and dies due to $d$.
Obviously, we have the connection
$$
\pi_{j,k}(t,d) = P(D^i_{j,k}=d|T^i_{j,k} = t)P(T^i_{j,k}=t)
$$
where $P(T_{j,k}^i=t)$ is given by $\lambda_{j,k}(t)$. We also denote by $\tilde{C}_{h(j)}$ the average screening cost per individual, given the policy $h$ determining whether the group $j$ is screened or not. This gives us the following expected costs.
\begin{proposition}
\label{prop:costs}
The total expected costs related to age group $j$ is given by
{\small\begin{equation}
\label{eq:expected-cost}
\mathbb{E}C_{j,h} = N_j\tilde{C}_{h(j)} + \sum_{k=0}^K N_j\mu_{j,h}(k)\sum_{t=0}^\infty \sum_{d=1}^D \tilde{C}_{j,k,t,d} \pi_{j,k}(t,d).
\end{equation}}
and the total expected costs during the entire follow up period is given by
\begin{equation}
\label{eq:expected-cost-total}
\mathbb{E}C_h = \sum_{j=1}^J \mathbb{E}C_{j,h}.
\end{equation}
\end{proposition}
\begin{proof}[Proof of Proposition \ref{prop:costs}]
Expected costs related to an individual $i$ of age $j$ and diagnosed with stage $k$ is given by
$$
\sum_{t=0}^\infty \sum_{d=1}^D \tilde{C}_{j,k,t,d} \pi_{j,k}(t,d).
$$
Hence the expected costs of the individual $i$ of age $j$ is obtained by conditioning on the stage $k$, leading to
$$
\E C^i_{j,h} = \sum_{k=0}^K \mu_{j,h}(k)\sum_{t=0}^\infty \sum_{d=1}^D \tilde{C}_{j,k,t,d} \pi_{j,k}(t,d).
$$
Adding the average screening cost $\tilde{C}_{h(j)}$ and multiplying with the number $N_j$ of individuals in age group $j$ leads to \eqref{eq:expected-cost}, from which the total costs \eqref{eq:expected-cost-total} follows by summing over age groups.
\end{proof}

\section{Data and parameter estimation}\label{data}

Our analysis relies on registry data from years 2000-2018. 
The data was created by linking the individual breast cancer data of Finnish women from the Finnish Cancer Registry (FCR data) with the breast cancer screening data from the Mass Screening Registry, a part of the FCR. We restricted to the data to diagnoses of first primary breast cancer, invasive and in situ carcinomas in breast (ICD-10 C50 \& D05). We also restricted the data to subsequent screening rounds for those aged at least 60 years and therefore excluded 4.6\% women (N=4226). After this exclusion we were able to assess a natural, long term (steady-state) screening of stage distribution. As total, 88607 women were included in the analysis. Cause of death (breast cancer, other cause) were included in the data. 
In the data, the cancer stages used by the FCR were coded further to five separate classes that were the stage 0 - Unknown, 1 - Localized, 2 - Non localized/Regional lymph nodes metastases, 3 - Metastasized farther than regional lymph nodes or invades adjacent tissues, and the stage 4 - In situ carcinoma.
The corresponding overall stage distributions were $9.6\%, 52.2\%, 25.4\%, 0.85\%,$ and $11.9\%$. Age-group specific incidence rates were extracted from NordCan \cite{NORDCAN} with some adjustments. 
For example, the incidence rate for the age group 50-51 year old females, that is the youngest screened age group in the current screening policy, was obtained from the Finnish Cancer Registry data. Note that NordCan provides data only on invasive cancer. Finnish Cancer Registry collects data on incidence of in situ carsinomas as well. Total age specific incidence rates for our analysis were obtained by combining data from NordCan and Finnish Cancer Registry. 

Stage specific survival rates were computed separately for all the age groups and the ones that were lost to follow-up were censored. We followed all age groups until the age of 99. We thus assumed that after 18 years of follow-up, patients did not have excess mortality due to their breast cancer and would survive as general female population in 2019. General population mortality rates were obtained from Statistics Finland. 

The costs of screening and the age and stage specific costs of treatment from specialized medical care were obtained from calculations for Lehtinen et al. (2019) \cite{Lehtinen}. In the data analysis, the unit costs of screening are 30 euros per invitee. 
The age and stage specific treatment costs of breast cancer corresponding to the first year after diagnosis, $C_1$, are displayed in Table \ref{tab:cost1}. The age and stage specific treatment costs of breast cancer corresponding to the years 2-5 after diagnosis, $C_2$,  are displayed in Table \ref{tab:cost2}. The age and stage specific treatment costs of breast cancer corresponding to the last year before breast cancer death, $C_3$, are displayed in Table \ref{tab:cost3}. 
If the patients dies during $n$ years after the breast cancer diagnosis for some other reason than breast cancer, the overall treatment costs are equal to:
\begin{equation}
\label{wams} C_n=
\begin{cases}
C_1+(n-1)C_2 &\text{for}\ 1\leq n\leq 5\\
C_1+4C_2 &\text{for}\ n > 5
\end{cases}
\end{equation}
If the patients dies from breast cancer during $n$ years after the breast cancer diagnosis, the overall treatment costs are equal to:
\begin{equation}
\label{wams2} C_n=
\begin{cases}
(n-1)C_1+C_3 &\text{for}\ n \in\{1,2\}\\
C_1+(n-2)C_2+C_3 &\text{for}\ n\in\{3,4,5\}\\
C_1+4C_2+C_3 &\text{for}\ n>5\\
\end{cases}
\end{equation}
Note that our approach to treatment costs is very conservative as in general breast cancer treatment and follow up lasts from five to ten years.

For assessing the effect of extending screening to the younger age groups, we extrapolated from the closest screened age groups. That is, we changed the stage distribution of the age group 46-47 old females to be the same as the stage distribution of the age group 50-51 old females under the current policy, the stage distribution of the age group 48-49 old females to be the same as the stage distribution of the age group 52-53 old females under the current policy, the stage distribution of the age group 50-51 old females to be the same as the stage distribution of the age group 52-53 old females under the current policy. In order to model the effect of the first screening, incidence rate of the age group 46-47 old females were increased by $28\%$ the incidence rate of the age group 48-49 old females were increased by $24.7\%$ and the incidence rate the age group 50-51 old females were decreased by $11.9\%$. 

When assessing the effect of extending screening to the older age groups, we assumed that the incidence with screening among elderly women would follow the pattern observed in a steady-state, in the neighbouring country Sweden. We shifted the decline in incidence to follow the pattern observed in Sweden where screening continues until the age of 74. Incidence rates in our model, under different policies, are displayed in Figure \ref{fig:incidencerates}. The stage distributions of age groups older than 69 year old females were changed such that changed the first two groups were changed to have the same stage distribution as the group 68-69 old females and then each older group is changed to have the stage distribution of the preceding groups. 
\begin{figure}[H]
  \includegraphics[width=\linewidth]{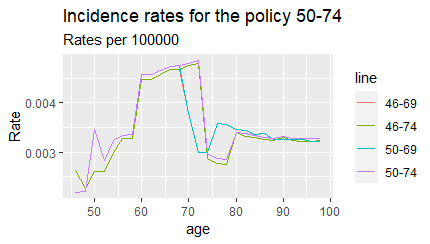}
  \caption{Incidence rates for different policies.}
\label{fig:incidencerates}
\end{figure}

In our results, we compare expected difference in costs per expected life years gained. That is, we compare
$$
\frac{\E C_{h_1}-\E C_{h_i}}{\E T_{h_1}-\E T_{h_i}},
$$
where $h_1$ is the current screening policy and $h_i$, $i=2,3,4$ is the modelled screening policy. A benefit of our model is that modelling of changes in the screening policy is affecting only the stage distribution $\mu$ (and through that to the dynamics of population size $N_j$). As such, survival probabilities and costs remain unaffected. In Table \ref{main}, we present the results under the four considered screening strategies for a cohort of 100000 individuals.

\section{Sensitivity analysis}
In addition to the main analysis, we conducted sensitivity analysis in order to assess the effect of the modelled incidence rates, estimated costs and the modelled stage distributions. The sensitivity analysis was conducted separately for the incidence rates, for the costs and for the stage distributions. In the sensitivity analysis, all the other factors were kept as in the main model. 

In assessing the effect of incidence rates we considered two separate scenarios. In the first one, we increased all the modelled incidence rates by 10\%. That is, when considering extending screening for the younger age group, we increased the incidence rates for the 46-69 year old women by 10\% while for 70 year old women and onwards the current incidence rates were used. When considering extending screening for the older age group, we increased the incidence rates by 10\% for women older that 69 years and the current incidence rates were used for 46-69 year old women. When considering extending to both, younger and older age groups, all incidence rates were increased. In the second scenario, the same was repeated using 10\% decrease. 

In assessing the effect of treatment costs we again considered two separate scenarios. In the first one, all treatment costs were increased by 10\%. In the second one, all treatment costs were increased by 50\%. 

In assessing the effect of the stage distribution, we considered a positive and a negative scenario. In the positive scenario all the modelled stage distributions were modified such that when considering extending screening for the younger age group, we kept the incidence rates as in the main analysis, but for 46-51 year old women the conditional probability (if diagnozed with cancer) for having localized cancer (stage 1) was increased by 0.02 and the conditional probability for having non-localized cancer (stage 2) was decreased by 0.02. When considering extending screening for the older age group, the conditional probability for having localized cancer was increased by 0.02 and the conditional probability for having non-localized cancer was decreased by 0.02 for women older that 69 years.  When considering extending to both, younger and older age groups, the modifications were done for the stage distributions of 46-51 year old women and of the women older than 69 years. In the negative scenario, the modifications were done to the other direction. That is, the conditional probabilities for having localized cancer (stage 1) was decreased by 0.02 and the conditional probability for having non-localized cancer (stage 2) was increased by 0.02 for the same age groups as in the positive scenario.

\begin{table*}[]
    \centering
    \begin{tabular}{|p{2cm}||p{2cm}|p{2cm}|p{2cm}|p{2cm}|p{2cm}|p{2cm}|}
 \hline
Target age group & Total costs (treatment and screening) & Total expected life-years & Ratio & Health care system incremental cost & Incremental cost-effectiveness ratio & Number of breast cancer deaths\\
 \hline
 50-69 yr & 245 498 112 & 386 0854,8 & 63,59 & & & 1686 \\
 46-69 yr & 246 267 889 & 386 1654,4 & 63,77 & 769 778 & 963 & 1658 \\
 50-74 yr & 253 498 611 & 386 0941,1 & 65,66 & 8000499 & 92736 & 1658 \\
 46-74 yr & 253 349 417 & 386 1865,8 & 65,60 & 7851305 & 7766 & 1616 \\
\hline
\end{tabular}
\caption{The effect of extending the breast cancer screening age group in Finland}
    \label{main}
\end{table*}
\begin{table*}[]
    \centering
    \begin{tabular}{|p{2cm}||p{2cm}|p{2cm}|p{2cm}|p{2cm}|p{2cm}|p{2cm}|}
 \hline
Target age group & Total costs (treatment and screening) & Total expected life-years & Ratio & Health care system incremental cost & Incremental cost-effectiveness ratio & Number of breast cancer deaths\\
  \hline
 50-69 yr & 245 498 112 & 386 0854,8 & 63,59 & & & 1686 \\
 46-69 yr & 261 211 188 & 385 9887,3 & 67,67 & 157 130 76 & -16242 & 1727 \\
 50-74 yr & 259 219 656 & 386 0237,3 & 67,15 & 137 215 44 & -22221 & 1743\\
 46-74 yr & 273 868 005 & 385 9408,4 & 70,96 & 283 698 93 & -19614 & 1769 \\
\hline
\end{tabular}
\caption{Sensitivity analysis, the effect of increasing the modelled incidence rates by 10\%}
    \label{inc1}
\end{table*}

\begin{table*}[]
    \centering
    \begin{tabular}{|p{2cm}||p{2cm}|p{2cm}|p{2cm}|p{2cm}|p{2cm}|p{2cm}|}
 \hline
Target age group & Total costs (treatment and screening) & Total expected life-years & Ratio & Health care system incremental cost & Incremental cost-effectiveness ratio & Number of breast cancer deaths\\
  \hline
 50-69 yr & 245 498 112 & 386 0854,8 & 63,59 & & & 1686 \\
 46-69 yr & 231 259 530 & 386 3428,3 & 59,86 & -142 385 82 & -5533 & 1588 \\
 50-74 yr & 247 760 910 & 386 1646,7 & 64,16 & 226 279 8 & 2857 & 1573 \\
 46-74 yr & 232 692 322 & 386 4338,8 & 60,22 & -128 057 90 & -3676 & 1462 \\
\hline
\end{tabular}
\caption{Sensitivity analysis, the effect of decreasing the modelled incidence rates by 10\%}
    \label{inc2}
\end{table*}

\begin{table*}[]
    \centering
    \begin{tabular}{|p{2cm}||p{2cm}|p{2cm}|p{2cm}|p{2cm}|p{2cm}|p{2cm}|}
 \hline
Target age group & Total costs (treatment and screening) & Total expected life-years & Ratio & Health care system incremental cost & Incremental cost-effectiveness ratio & Number of breast cancer deaths\\
 \hline
 50-69 yr & 267 210 206 & 386 0854,8 & 69,21 & & & 1686 \\
 46-69 yr & 267 456 097 & 386 1654,4 & 69,26 & 245 891 & 308 & 1658 \\
 50-74 yr & 275 500 860 & 386 0941,1 & 71,35 & 829 0654 & 96099 & 1658 \\
 46-74 yr & 274 735 067 & 386 1865,8 & 71,14 & 752 4862 & 7443 & 1616 \\
\hline
\end{tabular}
\caption{Sensitivity analysis, the effect  of increasing the treatment costs by 10\%}
    \label{cost1}
\end{table*}

\begin{table*}[]
    \centering
    \begin{tabular}{|p{2cm}||p{2cm}|p{2cm}|p{2cm}|p{2cm}|p{2cm}|p{2cm}|}
 \hline
Target age group & Total costs (treatment and screening) & Total expected life-years & Ratio & Health care system incremental cost & Incremental cost-effectiveness ratio & Number of breast cancer deaths\\
 \hline
 50-69 yr & 354 058 580 & 386 0854,8 & 91,70 & & & 1686 \\
 46-69 yr & 352 208 926 & 386 1654,4 & 91,21 & -184 9654 & -2313 & 1658 \\
 50-74 yr & 363 509 856 & 386 0941,1 & 94,15 & 945 1276 & 109553 & 1658 \\
 46-74 yr & 360 277 671 & 386 1865,8 & 93,29 & 6219091 & 6152 & 1616 \\
\hline
\end{tabular}
\caption{Sensitivity analysis, the effect  of increasing the treatment costs by 50\%}
    \label{cost2}
\end{table*}

\begin{table*}[]
    \centering
    \begin{tabular}{|p{2cm}||p{2cm}|p{2cm}|p{2cm}|p{2cm}|p{2cm}|p{2cm}|}
 \hline
Target age group & Total costs (treatment and screening) & Total expected life-years & Ratio & Health care system incremental cost & Incremental cost-effectiveness ratio & Number of breast cancer deaths\\
 \hline
 50-69 yr & 245 498 112 & 386 0854,8 & 63,59 & & & 1686 \\
 46-69 yr & 246 505 934 & 386 1579,1 & 63,84 & 100 7822 & 1391 & 1660 \\
 50-74 yr & 254 173 196 & 386 0806,9 & 65,83 & 867 5084 & -181278 & 1675 \\
 46-74 yr & 254 252 135 & 386 1658,4 & 65,84 & 875 4023 & 10893 & 1636 \\
\hline
\end{tabular}
\caption{Sensitivity analysis, the effect of decreasing the conditional probabilities of localized breast cancers}
    \label{stage1}
\end{table*}

\begin{table*}[]
    \centering
    \begin{tabular}{|p{2cm}||p{2cm}|p{2cm}|p{2cm}|p{2cm}|p{2cm}|p{2cm}|}
 \hline
Target age group & Total costs (treatment and screening) & Total expected life-years & Ratio & Health care system incremental cost & Incremental cost-effectiveness ratio & Number of breast cancer deaths\\
 \hline
 50-69 yr & 245 498 112 & 386 0854,8 & 63,59 & & & 1686 \\
 46-69 yr & 246 029 845 & 386 1729,8 & 63,71 & 531 733 & 607 & 1656 \\
 50-74 yr & 252 824 026 & 386 1075,2 & 65,48 & 732 5914 & 33239 & 1641 \\
 46-74 yr & 252 446 699 & 386 2073,2 & 65,37 & 694 8587 & 5703 & 1597 \\
\hline
\end{tabular}
\caption{Sensitivity analysis, the effect of increasing the conditional probabilities of localized breast cancers}
    \label{stage2}
\end{table*}

\begin{table*}[]
    \centering
    \begin{tabular}{|p{2cm}||p{2cm}|p{2cm}|p{2cm}|p{2cm}|p{2cm}|}
 \hline
Age group & $C_1$ Treatment cost (Stage 0) & $C_1$ Treatment cost (Stage 1) &  $C_1$ Treatment cost (Stage 2) & $C_1$ Treatment cost (Stage 3) & $C_1$ Treatment cost (Stage 4)\\
 \hline
 46-49 yr & 24800 &  20400& 28400 & 33300 & 15300
\\
 50-54 yr & 22400 & 18000 &26100 & 30900&12900
\\
 55-59 yr & 23300 & 18800 & 26900& 31800& 13800
\\
 60-64 yr & 20900&16400&24500&29400 & 11300
\\ 65-69 yr & 21100&16600 &24700 &29600 &11600
\\ 70-74 yr & 18700 &14300&22300&27200 &9200
\\75+ yr  & 14200 &9800&17800&22700 &4700
\\
\hline
\end{tabular}
\caption{Age and stage specific breast cancer treatment costs corresponding to the first year after diagnosis}
    \label{tab:cost1}
\end{table*}

\begin{table*}[]
    \centering
    \begin{tabular}{|p{2cm}||p{2cm}|p{2cm}|p{2cm}|p{2cm}|p{2cm}|}
 \hline
Age group & $C_2$ Treatment cost (Stage 0) & $C_2$ Treatment cost (Stage 1) &  $C_2$ Treatment cost (Stage 2) & $C_2$ Treatment cost (Stage 3) & $C_2$ Treatment cost (Stage 4)\\
 \hline
 46-49 yr & 4000&  2400& 3300& 6400& 2000
\\
 50-54 yr & 3700& 2000&2900 & 6000&1600
\\
 55-59 yr & 3400& 1700& 2600& 5700& 1300
\\
 60-64 yr & 3200 &1500&2400&5500& 1100
\\ 65-69 yr & 3400&1800 &2600&5700&1300
\\ 70-74 yr & 3200&1500&2400&5500&1100
\\75+ yr  & 3000&1300&2200&5300&900
\\
\hline
\end{tabular}
\caption{Yearly age and stage specific breast cancer treatment costs corresponding to the years 2-5 after diagnosis}
    \label{tab:cost2}
\end{table*}
\begin{table*}[]
    \centering
    \begin{tabular}{|p{2cm}||p{2cm}|p{2cm}|p{2cm}|p{2cm}|p{2cm}|}
 \hline
Age group & $C_3$ Treatment cost (Stage 0) & $C_3$ Treatment cost (Stage 1) &  $C_3$ Treatment cost (Stage 2) & $C_3$ Treatment cost (Stage 3) & $C_3$ Treatment cost (Stage 4)\\
 \hline
 46-49 yr & 38100&  33100& 38900 & 55400& 29200
\\
 50-54 yr & 31500 & 26600 &32400 & 48900&22700
\\
 55-59 yr & 27700& 22800 & 28500& 45000& 18900
\\
 60-64 yr & 24600&19700&25500&42000& 15800
\\ 65-69 yr & 26100&21200 &26900 &43400 &17300
\\ 70-74 yr & 16900 &11900&17700&34200 &8000
\\75+ yr  & 12300 &7300&13100&29600&3400\\
\hline
\end{tabular}
\caption{Age and stage specific breast cancer treatment costs corresponding to the last year before cancer death}
    \label{tab:cost3}
\end{table*}

\section{Results and Discussion}

Our analysis show that one should consider very carefully before extending any effective screening program to new age groups. For a cohort of 100000 individuals, extending the current screening strategy in Finland does not seem to provide large benefits, see Table \ref{main}. Current strategy in Finland is to provide biennial breast cancer screening for all 50-69 year old women. In our analysis, we considered extending screening to younger age groups starting from age 46 and to older age groups ending to age 74. In particular, based on our modeling, overall costs of extending screening to older age groups in Finland would be 92736 euros per a saved life year. In a cohort of 100000 individuals, the number of breast cancer deaths would decrease from 1686 to 1658. That is, the number of breast cancer deaths would decrease by 30 in the cohort's lifetime, 3.8 million women-years. Although extending breast cancer screening to older age groups would reduce mortality, the costs are high and screening may yield unnecessary worries and treatment related morbidities. Based on our modeling, the overall costs of extending screening to younger age groups in Finland would be 963 euros per a saved life year. The main reason for a much lower cost per a saved life year, when compared to extending screening to older age groups, is simply that the expected remaining life years are naturally much higher in young age groups. Based on our modelling, the decrease in the number of breast cancer deaths in a cohort of 100000 individuals, when extending screening to younger age groups, would be the same as when extending to the older age groups. Note, however, that our model does not consider screening related radiation burden that might be important on population level especially when screening starts from young age groups. Based on our model, if screening would be extended to both directions, the overall costs would be 7766 euros per a saved life year and in the cohort of 100000 individuals, the number of breast cancer deaths would decrease from 1686 to 1616. Note that incremental cost effectiveness ratio depend highly on incidence rates. In our analysis, we assumed that if we extend screening up to 74 year old women, the incidence rates of over 70 year old depend on the starting age of screening. If, under extending to older age groups only, we had used the same the same incidence rate for over 70 year old women as under extending to both directions, the incremental cost ratio for 50-74 would have been approximately 41 000 euros. Thus the small changes in incidence, displayed in Figure \ref{fig:incidencerates}, do have a large impact on the results. Note also that the estimations rely on extrapolation. The proposed model is suitable also for randomized data and when ever recent randomized data is available, it should be used instead of extrapolation.

Our sensitivity analysis reveals that if extended screening yield much higher incidence rates, extending screening is harmful, see Table \ref{inc1}. In that case screening would increase the costs but it would not lead to increase in expected life years. This highlights the importance of discussing the effects of increased incidence due to screening, overdiagnoses, and the harmful effects of unnecessary cancer treatment. Note, however, that based on our main model, the increase in incidence rates is much smaller that in the sensitivity analysis. Assessing the effect of unexpectedly high increase in incidence rates is, however, important before making new policies about screening target age groups. Note however, if the incidence rates are unexpectedly low, the benefits of screening increase, see Table \ref{inc2}. Under that scenario, extending screening to younger age group would both, lead to increase in expected life years and decrease in health care costs. Note that our sensitivity analysis related to incidence rates does not model the effect of changes in the actual incidence rates. Instead, it models the effect of estimating the changes in incidence rates incorrectly under the new screening policies.

Our sensitivity analysis related to the treatment costs reveals that if the treatment costs are increased by 10\%, then the overall costs of extending screening to younger age groups are only 308 euros per a saved life year and
if the treatment costs are increased by 50\%, then extending screening to younger age groups actually saves money. The reason behind that is that then the costs of additional screening in a cohort of 100000 individuals become lower that savings in treatment costs when cancers are detected earlier. If the treatment costs are increased by 10\%, then the overall costs of extending screening to older age groups are still high, 96099 euros per a saved life year, and if the treatment costs are increased by 50\%, then the overall costs of extending screening to older age groups are 109553 euros per a saved life year.

Our sensitivity analysis related to the effect of stage distribution reveals that if, assuming that cancer is diagnosed, the conditional probability of localised cancer increases, then the benefits of screening increase. If, assuming that cancer is diagnosed, the conditional probability of localised cancer decreases, then the benefits of screening decrease as well. The larger positive effect screening has on the stage distribution, the larger are the benefits of screening (assuming that the incidence rate does not increase too much).

When we applied the screening strategies to younger and older target ages, we made several assumptions. Several sensitivity analyses were performed to assess effect on results. In addition, we implicitly assumed that the attendance rate of the younger and older age groups would be the same as the attendance rate of the closest screened age groups. Since attendance to screening vary little by age group~ \cite{Screeningstatistics}, this assumption is can be regarded feasible and realistic for a Finnish cohort of 100,000 women with that respect. Often, the attendance rate is assumed to be $100\%$ in modeling studies and their results will thus represent an ideal and unrealistic world. We also assumed that the probability of detecting breast cancer in screening in younger and older ages is the same than in the closest screened age group which is the most feasible assumption. In addition, even if we used all Finnish breast cancer data, estimated survival probabilities may be uncertain which may affect our results. We are aware that cost-effectiveness models without natural history are likely to provide overly optimistic estimates. Nevertheless, magnitude and order of incremental costs-effectiveness ratios per life-years gained can be used to support policy-making.   

Extending screening saves lives, but it is expensive for older ages. The best strategy is to extend screening to younger ages with low additional incremental cost per life-year gained. Even then, however, number of averted breast cancer deaths is rather small compared to the current screening which prevents annually 100 breast cancer deaths each \cite{Annualreport2018}. Therefore, one possibility would be, instead of extending the target age groups, to extend screening to targeted groups. That is, to screen individuals whose risk for being diagnosed with breast cancer is exceptionally high. Note that the modeling approach provided in this paper can be applied for assessing the effects of risk-stratified screening assuming that suitable data is available. One could consider the cost effectiveness of screening individuals with known risk factors.


\subsection{Declaration of Conflicting Interests}
The authors declare no potential conflicts of interest with respect to the research, authorship, and/or publication of this article.

\subsection{Funding}
This work was supported by Cancer Foundation Finland Grant (grant number 200080).

\end{document}